\documentclass[11pt]{article}
\usepackage{amsthm}
\usepackage{amssymb}
\usepackage{graphicx}
\usepackage{amsmath}
\usepackage{verbatim}
\usepackage{setspace}
\usepackage{booktabs}
\usepackage{color}

\usepackage{tikz-cd}
\usepackage{mathtools}

\usepackage{subcaption}

\usepackage{tikz}
\usetikzlibrary{matrix}

\usepackage{hyperref}

\usepackage{ulem}

\newtheorem{assumption}{Assumption}
\newtheorem{proposition}{Proposition}
\newtheorem{remark}{Remark}

\newtheorem{example}{Example}

\usepackage[round]{natbib}

\def\bs{\boldsymbol}

\def\t{^{\top}}

\newcommand{\Perp}{\perp \! \! \! \perp}
\newcommand{\ve}{\varepsilon}

\begin{document}
{
\title{\textbf{Identification of Regression Models with a Misclassified and Endogenous Binary Regressor}
\author{Hiroyuki Kasahara\\
Vancouver School of Economics\\
University of British Columbia\\
hkasahar@mail.ubc.ca \and Katsumi Shimotsu\thanks{Address for correspondence: Katsumi Shimotsu, Faculty of Economics, University of Tokyo, 7-3-1, Hongo, Bunkyo-ku, Tokyo 113-0033, Japan. The authors are grateful to the co-editor, two anonymous referees, and participants of A Celebration of Peter Phillips' Forty Years at Yale Conference whose comments greatly improved the paper. This research is support by the Natural Science and Engineering Research Council of Canada and JSPS Grant-in-Aid for Scientific Research (C) No. 26380267.}\\
Faculty of Economics \\
University of Tokyo\\
shimotsu@e.u-tokyo.ac.jp
}}
\maketitle
}
\begin{abstract}
We study identification in nonparametric regression models with a misclassified and endogenous binary regressor when an instrument is correlated with misclassification error. We show that the regression function is nonparametrically identified if one binary instrument variable and one binary covariate satisfy the following conditions. The instrumental variable corrects endogeneity; the  instrumental variable must be correlated with the unobserved true underlying binary variable, must be uncorrelated with the error term in the outcome equation, but is allowed to be correlated with the misclassification error. The covariate corrects misclassification; this variable can be one of the regressors in the outcome equation, must be correlated with the unobserved true underlying binary variable, and must be uncorrelated with the misclassification error. We also propose a mixture-based framework for modeling unobserved heterogeneous treatment effects with a misclassified and endogenous binary regressor and show that treatment effects can be identified if the true treatment effect is related to an observed regressor and another observable variable.
\end{abstract}

\section{Introduction}

Misclassified endogenous binary regressors are prevalent in applications. Examples include self-reported educational attainment \citep{black03jasa}, self-reported participation in job training \citep{krueger98jle}, health insurance coverage reported by worker \citep{black00jasa} and participation to the Supplemental Nutrition Assistance Program (SNAP) formerly known as the Food Stamp Program \citep{kreider_etal12jasa}. For example, \citet{black03jasa} find that only 66.4\% of those reporting a professional degree in the 1990 Decennial Census have a professional degree, and \citet{meyer20jhr} find that 49.0\% of true food stamp recipient households do not report receipt in the Current Population Survey from 2002 to 2005.

We study identification in nonparametric regression models in the presence of a misclassified and endogenous binary regressor when a binary instrument controlling endogeneity is correlated with misclassification error. We consider the following model with a misclassified and endogenous binary regressor and the instrument variable (instrument) $Z$:
\begin{equation} \label{model0}
\begin{aligned}
Y &= g(X,  T^*) + \ve \\
& =  \alpha(X) + \beta(X)T^* + \varepsilon, \quad E[\varepsilon|X,Z]=0,
\end{aligned}
\end{equation}
where $Y$ is the outcome variable (for example, wage), $X$ is exogenous controls, and $\ve$ is an unobservable disturbance. $T^*$ is an unobservable binary regressor (for example, true educational qualification) which may be endogenous in the sense correlated with $\varepsilon$. $T$ is an observable misclassified measurement of $T^*$ (for example, self-reported schooling). Here, because the regressor $T^*$ is binary, its measurement error is necessarily nonclassical, i.e., $T^*-T$ is correlated with $T^*$. This makes identification difficult. 

A number of papers have studied regression models with an exogenous misclassified binary regressor. \citet{aigner73joe} characterizes the OLS asymptotic bias for such a model and develops a procedure to consistently estimate the coefficient of the misclassified binary regressor when the outside information on misclassification probabilities is available. More recently,  \citet{lewbel07em} shows that  the difference $E[Y|X,T^*=1]-E[Y|X,T^*=0]$ can be identified using an instrument that is mean independent of the change in outcome variable associated with the change in $T^*$ when the instrument takes at least three values. \citet{mahajan06em}  shows that the conditional mean of outcome variable $Y$ given $T^*$ is identified while  \citet{hu08joe} provides related identification results when the discrete regressor takes more than two values.    \citet{battistin14joe} examine the identification of the returns to educational qualifications when repeated misclassified measurements are available. \citet{black00jasa} and \citet{kaneetal99nber} show identification when repeated misclassified measurements of a binary regressor are available. 


Only a few papers analyze identification of regression models when a misclassified binary regressor is endogenous. In particular, \citet{mahajan06em} shows that $\alpha(X)$ and $\beta(X)$ are identified when there exists a binary instrument variable $Z$ that satisfies the conditional independence from $T$ given by
\begin{equation} \label{CI-TZ}
T \Perp Z \text{ conditionally on } (T^*,X), 
\end{equation}
in addition to the standard relevance condition and exclusion restriction as well as some other assumptions. However,
 \citet{ditraglia19joe} show that the assumptions in \citet{mahajan06em} imply that $E[\varepsilon|X,T^*]=0$, namely, $T^*$ is exogenous.  As a result, identification of the model (\ref{model0}) under endogenous $T^*$ has remained an open question.  
 
As pointed out by \citet{ditraglia19joe}, the reason why \citet{mahajan06em} cannot identify $\alpha(X)$ and $\beta(X)$ under endogenous $T^*$ is that \citet{mahajan06em} uses only one binary instrument $Z$ to control two sources of endogeneity, i.e., misclassification in $T$ and endogeneity in $T^*$.  
 
Some recent studies provide related identification conditions for models with an endogenous misclassified regressor $T^*$  while maintaining the assumption (\ref{CI-TZ}) that the instrument $Z$ is not only independent of $\varepsilon$ but also independent of $T$ conditional on $T^*$. 
\citet[][Theorem 2.3]{ditraglia19joe} provide the point identification of $\beta(X)$ under the higher-order independence assumption 
\begin{equation}\label{assn_ditraglia}
E[\varepsilon^2|X,Z] = E[\varepsilon^2|X],\  \ E[\varepsilon^3|X,Z] = E[\varepsilon^3|X],\  \text{and}\ E[\varepsilon^2|X,Z, T^*, T] = E[\varepsilon^2|X, Z, T^*].
\end{equation}
\citet{nguimkeu_etal19joe} analyze the  (local)  identification of a parametric model with endogenous treatment and endogenous misclassification using exclusion restrictions.  Their identification argument builds on that of \citet{poirier80joe}.  Both \citet[][Assumption 2.2.(i)]{ditraglia19joe} and \citet[][Assumption 1]{nguimkeu_etal19joe} assume  that the misclassification probability is not affected by the instrument $Z$ conditional on other observables. Other related studies include \citet{hu_etal15ej, hu_etal16el} who address the identification of nonseparable models with mismeasured endogenous regressor but their Assumption 2.1 also assumes that the instrument $Z$ is independent of $T$ conditional on $(T^*,X)$.  In these  studies, the instrument $Z$ has to satisfy two different exclusion restrictions: one from the outcome equation and the other from the misclassification probability. 

In empirical applications, a researcher chooses the instrument $Z$ such that $Z$ is relevant for $T^*$ and  is excluded  from the outcome equation; whether $Z$ is excluded from the misclassification probability or not is often a secondary concern given the difficulty of finding a valid instrument that satisfies both the relevance condition and the exclusion restriction from the outcome equation. When an endogenous binary regressor is a self-reported variable, however, the instrument $Z$ may be correlated with the misclassification error as the following examples illustrate.

\begin{example}[Supplemental Nutrition Assistance Program (SNAP)]
Many empirical papers analyze the effect of participation in SNAP on health outcomes or food insecurity, where $Y$ is a health outcome or food insecurity variable and $T$ is self-reported participation in SNAP. As the instrument $Z$ for controlling endogeneity in $T^*$, some studies use variables that affect the cost and benefit of participating in SNAP such as whether the state uses biometric identification technology (i.e., fingerprint scanning) and the percentage of SNAP benefits issued by direct mail rather than electronic benefit transfer (EBT) \citep[e.g.,][]{meyerhoefer08ajae, yen08ajae, almada16ajae}. For example, the use of EBT may help encourage SNAP participation by mitigating the stigma associated with SNAP participation \citep{yen08ajae}. Such state policies may be correlated with the misreporting error for the following two reasons.

First, as \citet[][Section 5.1]{boundetal01handbook} discuss, the salience of behavioral experience may affect the likelihood of misreporting by helping retrieval of the information from memory. In this context, the required use of fingerprint scanning may help people remembering that they participate in SNAP because the event of having their fingerprints taken may strengthen the memory trace. Similarly, receiving SNAP benefits by direct mailing may be more salient than receiving SNAP benefits by EBT because direct mailing requires visiting a bank and depositing a check. Therefore, the misreporting error can be correlated with state policies on the use of fingerprint scanning and the receipt of SNAP benefits by direct mail.  

Second, the correlation between misreporting error and state policies may arise because of unobserved heterogeneous psychological costs of participating in SNAP, or stigma. Stigma may induce SNAP recipients to make a false statement on their SNAP participation. As a result, the likelihood of misreporting may depend on the unobserved stigma level. On the other hand, true participation status is affected by both the use of EBT and stigma level, inducing the correlation between stigma level and the use of EBT conditional on true SNAP participation. As a result, the misreporting error may be correlated with the use of  EBT through the unobserved stigma level.\footnote{To see this, let $\omega \in\{0,1\}$ be stigma level. Then, $\Pr(T=1|T^*,\omega)\neq \Pr(T=1|T^*)$ if $\omega$ affects the likelihood of making a false statement. When both $\omega$ and $Z$ affect the true participation status $T^*$, by Bayes' theorem, $\Pr(\omega=1|T^*,Z)\neq \Pr(\omega=1|T^*)$ in general. Consequently, $\Pr(T=1|T^*,Z) = \sum_{\omega} \Pr(T=1|T^*,\omega)\Pr(\omega|T^*,Z)\neq \Pr(T=1|T^*)$.\label{footnote1}} Similarly, the misreporting error  may be correlated with the required use of fingerprint scanning if the true participation status is affected by both unobserved stigma and the required use of fingerprint scanning.  
\end{example}
 
\begin{example}[Returns to Education]
Consider analyzing returns to education on wages, where outcome variable $Y$ is logarithm of wage, $T^*$ is true educational attainment, and $T$ is self-reported educational attainment. Self-reported educational attainment may be misclassified when a respondent makes a false statement \citep{kaneetal99nber, battistin14joe}.  \cite{card93nber} proposes using college proximity as an instrument $Z$ for controlling endogeneity in $T^*$ given that college proximity affects the cost of attending colleges. College proximity could be correlated with the misreporting error of educational attainment for the following reasons. 

First, if a person spent her childhood in the same city or if her parents live in the same city as the location of the college she attended, then this makes it easier to retrieve the memory of attending colleges and reduces the chance of misreporting. Second, the correlation between misreporting error and college proximity can arise because of unobserved stigma of not obtaining a college degree. The likelihood of falsely reporting that one has obtained a college degree may depend on stigma levels associated with not obtaining a college degree. At the same time, stigma and college proximity may be correlated conditional on true educational attainment because the decision of obtaining a college degree is affected by both college proximity and  stigma. Then, as discussed in footnote \ref{footnote1}, the misclassification error could be correlated with college proximity  through unobserved stigma. 
\end{example}



To the best of our knowledge, none of the existing papers establishes identification of models with a misclassified endogenous binary regressor when an instrument is correlated with misclassification errors. This paper fills this gap. Specifically, we relax the assumption (\ref{CI-TZ}) and show identification when one of the covariates in the outcome equation, denoted by $V$, satisfies an exclusion restriction from the misclassification probability, i.e., 
\[
 T  \Perp V \text{ conditionally on } (T^*, X,Z), 
\] 
where the model  (\ref{model0}) is now written as 
\begin{equation} \label{model1}
\begin{aligned}
Y & =  \alpha(X,V) + \beta(X,V)T^* + \varepsilon, \quad E[\varepsilon|X, Z, V] = 0.
\end{aligned}
\end{equation}
Because $E[\varepsilon|X, Z, V]=0$, $V$ can be one of the covariates in the outcome equation. As in the existing literature, $V$ also needs to be relevant for $T^*$ in that $V$ changes the distribution of $T^*$. Unlike the existing literature, however, we allow $Z$ to affect the misclassification probability.

Choosing the variable $V$ in empirical applications may not be easy. One possibility is to refer to the existing studies that examine the determinant of misreporting errors, which may be survey-specific. 

\setcounter{example}{0}
\begin{example}[SNAP, continued]\cite{meyer20jhr} link administrative data on SNAP participation to the American Community Survey (ACS) and the Current Population Survey (CPS) in Illinois and Maryland and find high false negative rates of 33 and 49 percent in the ACS and CPS, respectively. They  examine how false negatives and false positives of self-reported SNAP participation are associated with individual's observed characteristics. Their estimates suggest that misreporting is not statistically associated with a gender dummy conditional on other observed characteristics.  Similarly, linking administrative data to the National Household Food Acquisition and Purchase Survey (FoodAPS), \cite{kang19soej}  find that neither false negative nor false positive is  statistically associated with a gender dummy. Given that gender is likely to be one of the key determinants of health outcome as well as SNAP participation, these results suggest that we may use a gender dummy for the variable $V$ when we use ACS, CPS, or FoodAPS to study the effect of SNAP participation on health outcome.
\end{example}

\begin{example}[Returns to Education, continued]
\cite{bingley14wp} link the Survey of Health, Ageing and Retirement in Europe (SHARE) to Danish Administrative Registers and examine the determinants of misclassification of self-reported schooling. They find that, conditional on true educational qualification and its interaction with income levels, misclassification error is correlated with neither gender nor age. This result suggests that gender and age are possible candidates for the variable $V$.
\end{example}

Figure \ref{fig1} compares the relationship among $Y$, $T^*$, $T$, $Z$, and $V$ in this paper with those in some recent studies.  Our  Proposition \ref{prop_1} in Figure \ref{fig1}(a) does not assume that  $Z$ is independent of $T$ conditional on $T^*$ while the existing studies such as  \citet{ditraglia19joe} and   \citet{nguimkeu_etal19joe} assume that $Z$ is excluded from the misclassification probability in Figure \ref{fig1}(b)(c).\footnote{\citet{nguimkeu_etal19joe} assume that the measurement error is not non-differential.} Figure \ref{fig1}(d) illustrates the approach of \citet{black00jasa}, \citet{kaneetal99nber}, and \citet{battistin14joe}, who use two conditionally independent measurements of $T^*$. In our setup, $Z$ and $V$ can be correlated to each other conditional on $T^*$ so that our identification argument is different from theirs.

\begin{figure}[h]
\begin{subfigure}{.47\textwidth}
  \centering 
\begin{tikzpicture}
  \matrix (m) [matrix of math nodes,row sep=3em,column sep=4em,minimum width=2em]
  {
    V  & & T \\
    Z &T^* &Y\\ };
  \path[-stealth]
    (m-1-1)   edge node [below] { } (m-2-2)  
            edge node [below] { } (m-2-3)  
    (m-2-1)  edge node [below] { } (m-2-2)
      edge[red]  node [below] { } (m-1-3)
    (m-2-2) edge node [left] { } (m-1-3)   
    edge node [left] { } (m-2-3)    ;
\end{tikzpicture} 
\caption{This paper's Proposition \ref{prop_1}: $Z$ can be correlated with $T$ conditional on $T^*$ } 
 \end{subfigure} 
  \begin{subfigure}{.47\textwidth}
  \centering 
\begin{tikzpicture}
  \matrix (m) [matrix of math nodes,row sep=3em,column sep=4em,minimum width=2em]
  {
     & & T \\
    Z &T^* &Y\\ };
  \path[-stealth]
    (m-2-1) 
            edge node [below] { } (m-2-2)
    (m-2-2) edge node [left] { } (m-2-3)  
            edge node [below] { } (m-1-3)   ;
\end{tikzpicture} 
\caption{ \citet{ditraglia19joe}: $Z$ is conditionally independent of  $Y$ up to the third moments}
 \end{subfigure}   
 \begin{subfigure}{.47\textwidth}
  \centering 
\begin{tikzpicture}
  \matrix (m) [matrix of math nodes,row sep=3em,column sep=4em,minimum width=2em]
  {
    V  & & T \\
    Z &T^* &Y\\ };
  \path[-stealth]
    (m-1-1) edge node [left] { } (m-1-3)  
            edge node [below] { } (m-2-2)  
            edge node [below] { } (m-2-3)  
    (m-1-3)  edge node [below] { } (m-2-3)
    (m-2-1)  edge node [below] { } (m-2-2)
    (m-2-2) edge node [left] { } (m-1-3)   
    		edge node [left] { } (m-2-3)    ;
\end{tikzpicture} 
\caption{\citet{nguimkeu_etal19joe} }   
 \end{subfigure}  
\begin{subfigure}{.5\textwidth}
  \centering 
\begin{tikzpicture}
  \matrix (m) [matrix of math nodes,row sep=3em,column sep=4em,minimum width=2em]
  {
     & &  T_1\\
     & &  T_2\\
    Z &T^* &Y \\};
  \path[-stealth]
    (m-3-1) 
            edge node [left] { } (m-3-2)
    (m-3-2) edge node   [left] { }  (m-1-3)  
     edge node   [left] { }  (m-2-3)  
        edge node [left] { }   (m-3-3) ;
\end{tikzpicture} 
\caption{Two repeated measures for $T^*$: \citet{black00jasa}, \citet{kaneetal99nber}, \citet{battistin14joe} }  
 \end{subfigure} 
 \caption{Comparison of the relationships between the outcome $Y$, true unobserved regressor $T^*$,  misclassified regressor $T^*$, instrument $Z$, and other covariate/instrument/measurement $V$ in some  published papers. Each arrow represents the dependence while a lack of arrow represents the conditional independence. } \label{fig1} 
\end{figure}
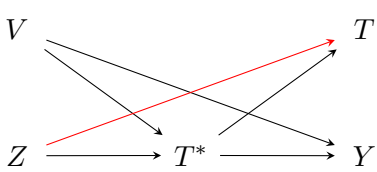
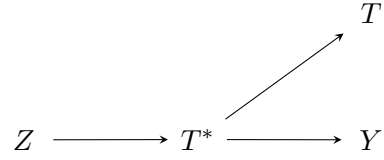
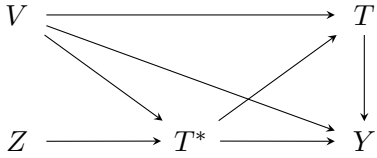
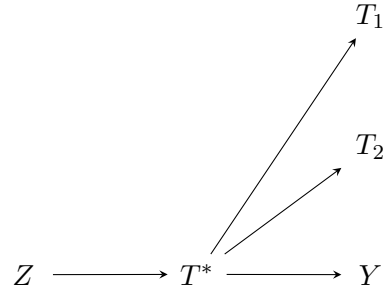

Our identification result is useful for empirical applications. To apply our identification result, the researcher needs to find one of the covariates that is correlated with endogenous regressor $T^*$ but does not affect misclassification. As in the examples above with SNAP and education qualification, the existing studies that link survey data to administrative data provide some guidance for choosing the covariate $V$. With such a covariate, we may weaken the requirement for the instrument $Z$ by allowing $Z$ to be correlated with the misclassification error.
 
The identification of the local average treatment effect (LATE) under mismeasured treatment was studied by \citet{yanagi19er}, \citet{ura18qe}, and  \citet{calvi19wp}. In his Assumption 4.3, \citet{yanagi19er} assumes that $V$ is excluded from both the outcome equation and the misclassification equation conditional on $T^*$, essentially giving an instrumental variable which shifts the  distribution of $T^*$  without affecting the outcome variable as well as the misclassification probability. \citet{ura18qe} obtains bounds for LATE under mismeasured treatment and standard LATE instrument assumptions. Using two different misclassified treatment indicators, \citet{calvi19wp} provide a point identification result for, what they call, the mismeasurement robust LATE which is generally different from the standard LATE. \citet{botosaru18jae} show that the average treatment effect on the treated is identifiable from repeated cross-section data when the treatment status is observed only either before or after the implementation of a program if there is a proxy variable for the latent treatment. \citet{Tommasi20} study identification and inference for the
bounds of the weighted average of local average treatment effects. 

The model (\ref{model1}) assumes that the individual treatment effect does not depend on unobservables. To examine heterogeneous treatment effects, we extend the model (\ref{model1}) by allowing $\alpha(\cdot)$ and $\beta(\cdot)$ to depend on an unobserved random variable $U^*$ that has a finite support. We generalize our identification result to this model with heterogeneous treatment effects when a mismeasured observable measure (proxy) for $U^*$ is available and show that the average treatment effect, the average treatment effect on the treated, the average treatment effect on the untreated, and the LATE are identified.

The remainder of this paper is organized as follows. Section 2 introduces the model and assumptions and derives identification results. Section 3 briefly discusses estimation and inference.  Section 4 shows identification of a heterogeneous treatment effect model.  Section 5 concludes.  Proofs are collected in Section 6.  All limits below are taken as $n \rightarrow \infty$. Let $: = $ denote ``equals by definition.'' For a $k\times 1$ vector $a$ and a function $f(a)$, let $\nabla_{a}f(a)$ denote the $k\times 1$ vector of the derivative $(\partial/ \partial {a}) f(a)$.

\section{Identification of the Model with a Misclassified Endogenous Binary Regressor}\label{sec:ident}

Throughout the paper, we assume that both $Z$ and $V$ are binary random variables with their support given by $\{0,1\}$. In this section, we suppress the exogenous regressor $X$ for brevity. The whole material remains valid conditional on $X$ if all the assumptions are imposed conditional on $X$.  We establish the identification of the model (\ref{model1}) under the following assumptions.
\begin{assumption}\label{assn_1} The following holds for any $(Z,V) \in \{0,1\}^2$ unless stated otherwise.
(a) $E[\ve|T^*,T,Z,V] = E[\ve|T^*,Z,V]$. (b) $E[\ve|T^*,Z,V] = E[\ve|T^*,V]$. (c) $\Pr(T^*=1|Z=0,V) \neq \Pr(T^*=1|Z=1,V)$. (d) $\Pr(T=1|T^*,Z,V) = \Pr(T=1|T^*,Z)$. (e) $\Pr(T^*=1|Z,V=0) \neq \Pr(T^*=1|Z,V=1)$. (f)  $0< \Pr(T=1|T^*=0,  Z)  < \Pr(T=1|T^*=1,  Z) <1$. (g) $\beta(V) \neq E[\ve|T^*=0,V]- E[\ve|T^*=1,V]$. (h) $0<\Pr(T^*=1|Z,V) <1$.
\end{assumption}
\begin{assumption}\label{assn_2} 
\begin{align*}
& \frac{\Pr(T^*=0|Z=0,V=0)\Pr(T^*=0|Z=1,V=1)}{\Pr(T^*=0|Z=1,V=0)\Pr(T^*=0|Z=0,V=1 )} \\
& \neq \frac{\Pr(T^*=1|Z=0,V=0)\Pr(T^*=1|Z=1,V=1)}{\Pr(T^*=1|Z=1,V=0)\Pr(T^*=1| Z=0,V=1 )}.
\end{align*}
\end{assumption}

\setcounter{example}{1}
\begin{example}[Returns to Education, continued]
In the return to education example, gender and age are possible candidates for the variable $V$. Gender and age are likely to satisfy Assumption \ref{assn_1}(d): $\Pr(T=1|T^*,Z,V,X) = \Pr(T=1|T^*,Z,X)$ because misclassification error is correlated with neither gender nor age conditional on true educational qualification and its interaction with income levels. Further, gender and age are likely to satisfy Assumption \ref{assn_1}(e): $\Pr(T^*=1|Z,V=0,X) \neq \Pr(T^*=1|Z,V=1,X)$ because these variable are correlated with true educational qualification conditional on other covariates.
\end{example}

Assumption \ref{assn_1} is a straightforward generalization of the assumptions in the current literature. Assumption \ref{assn_1}(a) assumes that the self-reported treatment status $T$ does not provide any additional information on the mean of $\ve$, and hence $Y$, given the knowledge of $T^*$, instrument $Z$, and covariate $V$ (and the exogenous regressor $X$). In particular, the error term $\ve$ is conditionally mean independent of the misclassification error conditional on $(T^*, Z,V)$. This is often referred to as ``non-differential measurement error."   While this assumption is standard in the misclassification literature (e.g.,  equation (1) of  \citet{mahajan06em} and Assumption 2.2.(iii) of \citet{ditraglia19joe}), it is potentially restrictive. In the context of the SNAP example,  this assumption may  be violated if misreporting of SNAP participation status due to unobserved stigma  is correlated with unobserved factors that affect  health outcome or food insecurity even after controlling for observed covariates  and true participation status. Similarly, in the returns to education example, this assumption may be violated if lying about college completion leads to higher earnings. In both cases, conditioning on a rich set of observed characteristics may mitigate the concern for the violation of this assumption. 

Assumption \ref{assn_1}(b) and (c) are the standard instrumental variable assumptions. Assumption \ref{assn_1}(b) states that the instrument $Z$ has to be excluded from the outcome equation while Assumption \ref{assn_1}(c) requires that $Z$ must be relevant for the true regressor $T^*$.    Combining Assumptions 1(a) and 1(b) gives $E[\ve|T^*,T,Z,V]=E[\ve|T^*,Z,V]=E[\ve|T^*, V]$.  Instrument validity, i.e., the validity of Assumption \ref{assn_1}(b)(c), is an important issue in empirical applications.

In Assumption \ref{assn_1}(d), we relax one requirement on the instrumental variable $Z$ in the existing misclassification literature by allowing $Z$ to be correlated with the misclassification probability, thus relaxing Assumption 3 in \citet{mahajan06em} and Assumption 2.2(i) in \citet{ditraglia19joe}. At the same time, Assumptions \ref{assn_1}(d) and \ref{assn_1}(e) require the existence of a covariate $V$ that affects the true regressor $T^*$ but does not affect misclassification error given $Z$ and other covariates. As discussed in the above examples,  existing studies linking administrative and survey data analyze how the misclassification error is associated with observed covariates, providing some guidance on how to choose $V$ from a set of observed covariates.  For instance, in the case of SNAP participation, a gender dummy can be used for $V$ because it may not be correlated with misclassification error in SNAP participation once other covariates are conditioned on but is likely to be correlated with true SNAP participation status $T^*$. 

Assumption \ref{assn_1}(f) requires that $T^*$ changes the mean of $T$ and corresponds to Assumption 2 in \citet{mahajan06em} and Assumption 2.2(ii) in \citet{ditraglia19joe}. Assumption \ref{assn_1}(f)  holds when $\Pr(T=1|T^*=1,  Z)>1/2$ and $\Pr(T=0|T^*=0)>1/2$, i.e., the value of $T^*$ is informative on the value of $T$.

Assumption \ref{assn_1}(g) requires that $T^*$ changes the conditional mean $E[Y|T^*,V]$ because taking the conditional expectation of model (\ref{model1}) conditional on $(T^*,V)$ gives
\begin{equation}\label{mean_Y}
E[Y|T^*=1,V] - E[Y|T^*=0,V] = \beta(V) + E[\ve|T^*=1,V] - E[\ve|T^*=0,V] \neq 0.
\end{equation}
In the SNAP example, this assumption holds if the conditional mean of health outcome given covariates differs between the SNAP recipients and non-recipients. Assumption \ref{assn_1}(h) holds if model (\ref{model1}) is nontrivial; if this assumption is violated, then there is no variation in the treatment status and identifying the treatment effect is impossible.

Assumption \ref{assn_2} holds if the changes in $(Z,V)$ induce sufficient variation in $\Pr(T^*=0|Z,V)$. A necessary condition for Assumption \ref{assn_2} is that both $Z$ and $V$ are relevant for $T^*$, namely, Assumptions \ref{assn_1}(c) and \ref{assn_1}(e) hold. 
In the SNAP example with gender dummy as $V$, the relevance of $V$ for $T^*$ holds if the true SNAP participation probability differs between male and female conditional on instrument $Z$ and other exogenous covariates $X$.

Under Assumption \ref{assn_1}(a)(b)(d), we obtain the following decomposition of $E[Y|Z,V]$, $E[T|Z,V]$ (see the proof of Proposition \ref{prop_1} for derivation):
\begin{equation}\label{system}
\begin{aligned}
E[Y|Z,V] & = E[Y|T^*=0,V] \Pr(T^*=0|Z,V)  + E[Y|T^*=1,V]\Pr(T^*=1|Z,V),   \\
E[T|Z,V] & = E[T|T^*=0,Z]\Pr(T^*=0|Z,V)  + E[T|T^*=1,Z]\Pr(T^*=1|Z,V), \\
E[YT|Z,V] & =  E[Y|T^*=0,V] E[T|T^*=0,Z] \Pr(T^*=0|Z,V) \\
& \quad +  E[Y|T^*=1,V] E[T|T^*=1,Z] \Pr(T^*=1|Z,V). 
\end{aligned}
\end{equation}
Evaluating them at $(Z,V) \in \{0,1\}^2$ gives 12 equations for 12 unknowns $\{E[Y|T^*=0,V],E[Y|T^*=1,V], \Pr(T^*=1|Z,V) , E[T|T^*=0,Z], E[T|T^*=1,Z]: (Z,V) \in \{0,1\}^2 \}$. Assumption \ref{assn_1}(e)--(g) and Assumption \ref{assn_2} enable us to solve these equations for a unique solution. When $Z$ is uncorrelated with $T^*$,  Assumption \ref{assn_2} does not hold and the system of equations (\ref{system}) fails to have a unique solution. $\alpha(V)$ and $\beta(V)$ are identified from the relation
\begin{equation}\label{alpha_beta}
\begin{pmatrix}
\alpha(V) \\
\beta(V)
\end{pmatrix}
=
\begin{pmatrix}
1 & \Pr(T^*=1|Z=0,V) \\
1 & \Pr(T^*=1|Z=1,V) 
\end{pmatrix}^{-1}
\begin{pmatrix}
E[Y|Z=0,V] \\
E[Y|Z=1,V] 
\end{pmatrix},
\end{equation}
and Assumption \ref{assn_1}(c).

The following proposition provides the main identification result of this paper. 

\begin{proposition}\label{prop_1}
Suppose that Assumptions \ref{assn_1} and \ref{assn_2} hold. Then, $\alpha(V)$, $\beta(V)$, $\Pr(T=1|T^*,Z)$, and $\Pr(T^*=1|Z,V)$ are identified for all $(T^*,Z,V)$.
\end{proposition} 


The key assumption in Proposition \ref{prop_1} that is different from those in the existing papers is that we allow $Z$ to affect the misclassification probability. \citet{mahajan06em} and \citet{ditraglia19joe} use one instrument $Z$ and assume $Z$ is independent of $T$ conditional on $T^*$. Our identification condition relaxes the requirement for $Z$ in the existing studies by alternatively assuming that one of the covariates in the outcome equation satisfies the exclusion restriction from the misclassification probability.

For clarification, we make the following two remarks. 
\begin{remark}
Proposition \ref{prop_1} holds even if $E[T|T^*,Z]$ does not depend on $Z$. Hence, $Z$ may or may not affect the misclassification probability.
\end{remark} 
\begin{remark} 
Proposition \ref{prop_1} holds even if $\alpha(V)=\alpha$ and $\beta(V)=\beta$. That is, the variable $V$ may or may not be one of the covariates in the outcome equation. 
\end{remark} 

We also consider an alternative set of assumptions in which $V$ does not satisfy the relevance condition, namely, $\Pr(T^*=1|Z,V)$ does not depend on $V$. Even in this case, we can identify the model if  the misclassification probability does not depend on $Z$. 

\begin{assumption} \label{assn_3}
Assumption \ref{assn_1} holds except for part (e), $\Pr(T^*=1|Z,V) = \Pr(T^*=1|Z)$ for all $Z$, and Assumption \ref{assn_1}(d) is strengthened to $\Pr(T=1|T^*,Z,V) = \Pr(T=1|T^*)$ for all $(Z,V)$.
\end{assumption}
\begin{proposition}\label{prop_2}
Under Assumption \ref{assn_3}, $\alpha(V)$, $\beta(V)$, $\Pr(T=1|T^*)$, and $\Pr(T^*=1|Z)$ are identified for all $(T^*,Z,V)$.
\end{proposition} 

Similarly to  \citet{ditraglia19joe} and other existing studies, Proposition \ref{prop_2} assumes that $Z$ is excluded from the misclassification probability. Proposition \ref{prop_2} shows that  the regression coefficient $\beta(V)$ can be identified  if we have a covariate $V$ that is excluded from the misclassification probability but can be  irrelevant for $T^*$, complementing the identification result of  \citet{ditraglia19joe} by providing an alternative assumption  to  the higher-order independence assumption of \citet{ditraglia19joe}  in equation (\ref{assn_ditraglia}).



\section{Estimation and Inference of Model Parameter}

In this section, we briefly discuss estimation and inference of the model parameter. Our estimation strategy follows directly from the identification result in Section \ref{sec:ident}. Suppose we have iid observations $\{(Y_i,T_i,X_i,Z_i,V_i): i=1, \ldots,n \}$ satisfying model (\ref{model1}) and Assumptions \ref{assn_1} and \ref{assn_2}. For a given value $x$ of $X$, let $\theta_x :=\{\alpha(x,v),\beta(x,v), \Pr(T^*=1|x,z,v), E[T|T^*=0,x,z], E[T|T^*=1,x,z] : (z,v) \in \{0,1\}^2 \}$ denote the vector of model parameter, and let $\phi_x : = \{E[Y|T^*=0,x,v], E[Y|T^*=1,x,v],  \Pr(T^*=1|x,z,v), E[T|T^*=0,x,z], E[T|T^*=1,x,z] : (z,v) \in \{0,1\}^2 \}$ denote the vector of 12 unknowns of the system (\ref{system}). Note that $\theta_x$ is obtained by replacing $\{E[Y|T^*=0,x,v], E[Y|T^*=1,x,v]: v \in \{0,1\} \}$ in $\phi_x$ with $\{\alpha(x,v),\beta(x,v): v \in \{0,1\} \}$. Let $m_x := \{E[Y|x,z,v], E[T|x,z,v], E[YT|x,z,v]: (z,v) = \{ (0,0), (1,0), (0,1), (1,1) \} \}$ denote the $(12 \times 1)$-vector of population conditional moments of $Y$ and $T$ conditional on $(X,Z,V)$ evaluated at $X=x$ and the support of $(Z,V)$.

Write the system (\ref{system}) of 12 equations as $m_x = f(\phi_x)$. Because $\phi_x$ is uniquely identified from population moments, $\phi_x$ is the unique solution that satisfies $m_x = f(\phi_x)$. Further, equation (\ref{alpha_beta}) gives $\alpha(x,v)$ and $\beta(x,v)$ as a function of $\phi_x$ and $m_x$; consequently, we may write $\theta_x$ as $\theta_x = g(\phi_x, m_x)$ for a smooth function $g$. Let $\widehat m_x$ be an estimator of $m_x$. We provide the details of the construction of $\widehat m_x$ later. We estimate $\phi_x$ and $\theta_x$ by $\widehat \phi_x := \arg\min_{\phi}\| \widehat m_x - f(\phi)\|^2$ and $\widehat \theta_x := g(\widehat \phi_x, \widehat m_x)$. A straightforward application of minimum distance estimation \citep[][Theorem 3.2]{neweymcfadden94hdbk} and the delta-method gives the following proposition:
\begin{proposition} \label{prop_distn}
Suppose $\widehat m_x$ satisfies $a_n(\widehat m_x-m_x)\to_d N(0,\Omega)$ for a nonstochastic sequence $a_n \to \infty$. Then, we have
\[
a_n (\widehat \phi_x -\phi_x) \to_d N(0, F^{-1}\Omega (F^{-1})^{\top}), \quad a_n (\widehat \theta_x -\theta_x) \to_d N(0, B \Omega B^{\top}),
\]
where $F := \nabla_{\phi^{\top}} f(\phi_x)$ and $B:=\nabla_{\phi^{\top}} g(\phi_x, m_x) F^{-1} + \nabla_{m^{\top}} g (\phi_x, m_x)$.
\end{proposition}
 
We proceed to the construction of $\widehat m_x$. Let $W := (Z,V)^{\top}$, $w_1:=(0,0)^{\top}, w_2:=(1,0)^{\top}, w_3:=(0,1)^{\top}, w_4:=(1,1)^{\top}$, and $R:=(Y,T,YT)^{\top}$. Then, the vector of conditional moments $m_x$ is written as $m_x = (E[R|x,w_1]^{\top},E[R|x,w_2]^{\top},E[R|x,w_3]^{\top},E[R|x,w_4]^{\top})^{\top}$. When $X$ has a finite support, we estimate $m_x$ by sample moments. For $j=1,\ldots,4$, define 
\[
\left. \widehat m_{xfj} : = \sum_{i=1}^n R_i I\{X_i=x\} I\{W_i = w_j \} \middle/ \sum_{i=1}^n I\{X_i=x\} I\{W_i = w_j \} \right. ,
\]
and define $\widehat m_{xf} : = (\widehat m_{xf1}^{\top}, \widehat m_{xf2}^{\top},\widehat m_{xf3}^{\top},\widehat m_{xf4}^{\top})^{\top}$. From Theorems 3.3.1 and 3.3.2 of \citet{bierens87book}, we obtain
\[
\sqrt{n}(\widehat m_{xf} - m_x) \to_d N(0,\Omega_f), \quad
\Omega_f = \begin{pmatrix}
\Omega_{f1} & 0 & \ldots & 0 \\
0 & \Omega_{f2} & \ddots & \vdots\\
\vdots &  \ddots & \Omega_{f3} & 0 \\
0 & \ldots &  0& \Omega_{f4}
\end{pmatrix},
\]
where $\Omega_{fj} := \text{Var}[R|x,w_j]$. Applying this result to Proposition \ref{prop_distn} with $a_n = \sqrt{n}$ and $\Omega = \Omega_f$ gives the asymptotic distribution of $\widehat \theta_x$.

When $X$ is continuously distributed, we estimate $m_x$ by a kernel estimator of $E[R|x,w]$ following  \citet{mahajan06em}, provided that $E[R|x,w]$ is continuous in $x$. For $j=1,\ldots,4$, define 
\[
\left. \widehat m_{xcj} : = \sum_{i=1}^n R_i K\left(\frac{X_i-x}{h}\right) I\{W_i = w_j \} \middle/ \sum_{i=1}^n K\left(\frac{X_i-x}{h}\right) I\{W_i = w_j \} \right.,
\]
where $K(\cdot)$ is a kernel function, and $h$ is the bandwidth satisfying $h + 1/(nh^{\text{dim}(X)}) \to 0$. Define $\widehat m_{xc}$ similarly to $\widehat m_{xf}$. Let $f(x|w)$ denote the density of $X$ conditional on $W=w$. Suppose that $E[R|x,w_j]$, $\text{Var}[R|x,w_j]$, $f(x|w)$, $K(\cdot)$, and $h$ satisfy Assumptions 10--14 of \citet{mahajan06em}. Then, it follows from Lemma 2 of \citet{mahajan06em} (see also Theorem 3.2.1 of \citet{bierens87book}) that
\[
\sqrt{nh}(\widehat m_{xc} - m_x) \to_d N(0,\Omega_c), \quad
\Omega_c = \begin{pmatrix}
\Omega_{c1} & 0 & \ldots & 0 \\
0 & \Omega_{c2} & \ddots & \vdots\\
\vdots &  \ddots & \Omega_{c3} & 0 \\
0 & \ldots &  0& \Omega_{c4}
\end{pmatrix},
\]
where $\Omega_{cj} := \text{Var}[R|x,w_j] \int K(s)^2 ds / f(x|w_j) \Pr(W=w_j)$. Therefore, applying this result to Proposition \ref{prop_distn} with $a_n = \sqrt{nh}$ and $\Omega = \Omega_c$ gives the asymptotic distribution of $\widehat \theta_x$.

\section{Heterogeneous Treatment Effect}


In Section \ref{sec:ident}, we assume that the effect of $T^*$ on $Y$ does not depend on unobservables. In this section, we extend the model (\ref{model1}) to allow the parameter $\alpha(\cdot)$ and $\beta(\cdot)$ to depend on an unobserved random variable $U^*$. This gives a random coefficient model similar to the model in \citet{heckman_etal06rest}:
\begin{equation}   \label{model_hetero} 
Y   = \alpha(U^*,X,V) + \beta(U^*,X,V)T^* + \varepsilon, \quad E[ \varepsilon|U^*,X,Z,V]=0,
\end{equation}
where $U^*$ is assumed to be exogenous but $T^*$ may be correlated with $\ve$. We allow $U^*$ and $T^*$ to be correlated. Hence, $\alpha(U^*,X,V)$ and $\beta(U^*,X,V)$ may be correlated with $T^*$ conditional on $(X,V)$. When $\alpha(U^*,X,V)$ and $T^*$ are correlated, we have ``sorting on the level,'' which is a common form of selection bias. When $\beta(U^*,X,V)$ and $T^*$ are correlated, we have ``sorting on the gain,'' which is called essential heterogeneity by \citet{heckman_etal06rest}. 

\begin{example}[Heterogeneous Effect of SNAP]\label{example:snap_hetero}
\citet{deb18el} examine how SNAP participation affects food insecurity while allowing for heterogeneous treatment effects using finite mixture models with two latent classes; low and high food security latent classes.  In this case, $Y$ is an ordinal measure of food security taking an integer value between $0$ and $10$. $U^*$ represents the latent variable and takes the value $0$ for a low food security class and $1$ for a high food security class such that $\alpha(0,X,V) <\alpha(1,X,V)$  for any value of $V$. Using the CPS and assuming no measurement error in self-reported SNAP participation, they find that the effect of SNAP participation is higher for a low food security class than for a high food security class. This suggests a possibility that $\beta(0,X,V)>\beta(1,X,V)$. 
\end{example} 

\begin{example}[Heterogeneous Returns to Education]
As in \citet{carneiro_etal11aer}, consider a model of heterogeneous returns to education, where  $Y$ is logarithm of wage, $U^*$ may be interpreted as unobservable ability, $T^*$ is true educational attainment (college degree), $V$ may be gender, and $Z$ is proximity to colleges. In wage regression (\ref{model_hetero}), unobserved ability $U^*$ may affect not only the intercept, $\alpha(U^*,X,V)$, but also the returns to schooling, $\beta(U^*,X,V)$. In such a case, $\alpha(U^*,X,V)$ and $\beta(U^*,X,V)$ may be correlated with $T^*$ conditional on $(X,V)$ because a person with higher unobserved ability may have a higher chance of  obtaining a college degree and because a forward-looking school decision depends on the returns to schooling.  \end{example} 

Both \citet{deb18el}  and \citet{carneiro_etal11aer} in the above examples assume no measurement error in treatment variables.\footnote{\citet{heckman_etal06rest} and \citet{HeckmanVytlacil07handbook} examine the case where $T^*$ is observable and develop procedures to estimate the summary statistics for $\beta(U^*,X,V)$ via the marginal treatment effect.  
} We relax this assumption by assuming that we have an observable binary measurement $T$ for an unobserved binary treatment variable $T^*$. To identify the joint distribution of $T^*$ and $U^*$ from the data, we augment the model with an  observable measurement $U$ of $U^*$. The role of $U$ to $U^*$ is similar to that of $T$ to $T^*$ in that  $U$ provides information on $U^*$.  

\setcounter{example}{2}
\begin{example}[Heterogeneous Effect of SNAP, continued]
\citet{deb18el} also analyze the determinants of latent food security class membership. They find that household size and a subjective variable indicating whether the individual met food needs or not are significant determinants of class membership. This suggests that these variables may be used for $U$.
\end{example}

\begin{example}[Heterogeneous Returns to Education, continued]
 To study the returns to education in the United States using the National Longitudinal Survey of Youth (NLSY), the existing studies use 
  the Armed Services Vocational Aptitude Battery (ASVAB) test scores to measure the cognitive  skills \citep[e.g.,][]{castex14jol, heckman18jpe}.  The measures such as ASVAB test scores can be used for $U$ as a proxy for unobserved ability $U^*$.
\end{example}

Henceforth, we suppress the exogenous regressor $X$ for brevity. The whole material remains valid conditional on $X$ if all the assumptions are imposed conditional on $X$. Define $S := (U,T)$ and $S^*:=(U^*,T^*)$. We assume that $S$ is conditionally independent of $V$ given $(S^*,Z)$, where $V$ is a binary observable variable. As in Section \ref{sec:ident}, we may choose $V$ based on the existing studies that use administrative data while we may choose the use of fingerprint technology or EBT as an instrument $Z$ for SNAP participation. As shown in Proposition \ref{prop_3} below, with additional regularity conditions (rank conditions and distinct eigenvalues), we may identify $\alpha(U^*,V)$, $\beta(U^*,V)$, $\Pr(S^*|Z,V)$, and $\Pr(S|S^*,Z)$ for all $(S,S^*,Z,V)$. Furthermore, the conditional distribution of $Y$ conditional on $(S^*,V)$ is identified.

We  assume that both  $U^*$ and $U$ take  $K_u$ discrete values with the support $\mathcal{U}:=\{u_1,\ldots,u_{K_u}\}$. When $U$ is a continuous variable as in the case of ASVAB, we may define $K_u$ distinct sets  by partitioning the support of $U$. Denote the support of $S^*$  and $S$ by $\mathcal{S}:=\{s_1,\ldots,s_K\}$ with $K:=2K_u$.
We also assume that  $Y$ can take at least  $K$ different values.


\begin{assumption}\label{assn_1h} (a) $\ve$ is independent of $S$ conditional on $(S^*,Z,V)$. (b) $\ve$ is independent of $Z$ conditional on $(S^*,V)$. (c) $\Pr(T^*=1|U^*,Z=0,V) \neq \Pr(T^*=1|U^*,Z=1,V)$ for any $U^* \in \mathcal{U}$ and $V \in \{0,1\}$. (d) $S$ is independent of $V$ given $(S^*,Z)$. (e) $\Pr(S=s|S^*=s,Z) > \Pr(S=s'|S^*=s,Z)$ for any $s' \neq s$ with $s,s' \in \mathcal{S}$ and for any $Z\in \{0,1\}$. (f) $\Pr(S^*=s|Z,V)>0$ for any $s \in \mathcal{S}$ and $(Z,V) \in \{0,1\}^2$.
\end{assumption} 
  

\begin{assumption}\label{assn_2h}
$\{\tilde \lambda_j\}_{j=1}^{K}$ take distinct values across $j=1,\ldots,K$, where
\[
\tilde \lambda_j:= \frac{\Pr(S^*=s_j|Z=0,V=0)\Pr(S^*=s_j|Z=1,V=1)}{\Pr(S^*=s_j|Z=1,V=0)\Pr(S^*=s_j|Z=0, V=1)}.
\]
\end{assumption}

\begin{assumption}\label{assn_3h}
There exists a partition of
the support of  the distribution of $Y$, $\{\Delta_j\}_{j=1}^K$, such that the matrix
\begin{align*}
& L_Y( V) := \\
& \begin{bmatrix}
1&1&\cdots &1\\
\Pr(Y\in \Delta_1|S^*=s_1,V)&\Pr(Y\in \Delta_1|S^*=s_2,V)&\cdots& \Pr(Y\in \Delta_1|S^*=s_{K},V)\\
\vdots  & \vdots& \ddots& \vdots  \\
\Pr(Y\in \Delta_{K-1}|S^*=s_1, V)&\Pr(Y\in \Delta_{K-1}|S^*=s_2,V)&\cdots& \Pr(Y\in \Delta_{K-1}|S^*=s_{K},V)
\end{bmatrix}
\end{align*}
is nonsingular for any $V \in \{0,1\}$.
\end{assumption}

Assumption \ref{assn_1h}(a)(b) corresponds to Assumption \ref{assn_1}(a)(b), representing  a non-differential measurement error assumption and an exclusion restriction on $Z$ from the outcome equation. Assumption \ref{assn_1h}(c) corresponds to Assumption \ref{assn_1}(c) and requires that $Z$ must be relevant for the true regressor $T^*$ at any value of $U^*$. Assumption \ref{assn_1h}(d) requires that $V$ is excluded from the measurement equation for $S$ conditional on  $(S^*,Z)$, generalizing Assumption \ref{assn_1}(d). As discussed through examples, the result of the existing studies that examine the determinant of misclassification errors may provide guidance on the choice of $V$.

Assumption \ref{assn_1h}(e) corresponds to Assumption \ref{assn_1}(f) and assumes that   $S$ is sufficiently informative to identify the unobserved value of $S^*$ such that the probability of $S=s$ given $S^*=s$ is higher than that of $S=s'$ for any $s'\neq s$. In Example \ref{example:snap_hetero}, the variable $S$ may consist of a self-reported SNAP participation $T$ and a binary subjective variable for meeting the food needs $U$. Then, Assumption \ref{assn_1h}(d) requires that the probability of truthfully reporting SNAP participation and food needs is larger than the probability of falsely self-reporting any combination of SNAP participation and food needs. Assumption \ref{assn_1h}(f) corresponds to Assumption \ref{assn_1}(h).

Assumption \ref{assn_2h} is similar to Assumption \ref{assn_2} and requires that $Z$ and $V$ are relevant for determining $\Pr(S^*|Z,V)$ and the changes in $(Z,V)$ induce sufficient variation in $\Pr(S^*|Z,V)$. Assumption \ref{assn_3h} generalizes Assumption \ref{assn_1}(g), requiring that the distribution of $Y$ changes sufficiently across different values of $(U^*,T^*)$ given $V$. In Example \ref{example:snap_hetero}, this assumption holds if the conditional distributions of an ordinal measure of food security given other covariates are sufficiently different across different SNAP participation statuses and latent food security classes.
 
\begin{proposition}\label{prop_3}
Under Assumptions \ref{assn_1h}--\ref{assn_3h}, $\alpha(U^*,V)$, $\beta(U^*,V)$, $\Pr(S|S^*,Z)$, $\Pr(S^*|Z,V)$, and $\Pr(Y\in \Delta|S^*,V)$ are identified for all $(S,S^*,Z,V)$ and for any set $\Delta$ on the support of the distribution of $Y$.
\end{proposition} 

 
We consider identification of treatment effects from model (\ref{model_hetero}). The local average treatment effect is the average of the treatment effect on $Y$ over the subpopulation (the compliers) whose treatment status is strictly affected by the instrument. If their Conditions 1 and 2 hold conditional on $V$,  \citet[][Theorem 1]{imbensangrist94ecma} show that the local average treatment effect equals
\[
\frac{E[Y| Z=1,V] - E[Y| Z=0,V]}{E[T^*| Z=1,V] - E[T^*|  Z=0,V]}.
\]  
This can be identified from Proposition \ref{prop_3} because $E[Y| Z,V]  = \sum_{s\in \mathcal{S}}  E[Y|S^*=s,  V] \Pr(S^*=s|Z)$ and $E[T^*| Z,V]  =  \Pr(T^*| Z,V)$. 

For identifying other treatment effects, let $Y_1$ denote the potential outcome if the subject were to receive treatment and let $Y_0$ denote the potential outcome if the subject were not to receive treatment. Decompose $Y_j$ into its conditional mean given $V$, $\mu_j(V)$, and its deviation from the mean, $\eta_j$, as 
\[
Y_1 = \mu_1(V) + \eta_1, \quad Y_0 = \mu_0(V) + \eta_0.
\]
We consider the following assumption to identify treatment effects.
\begin{assumption}\label{assn_treatment}
Either $E[\eta_0|U^*,V] = E[\eta_0| U^*,Z,V]$ or $E[\eta_1| U^*,V] = E[\eta_1| U^*,Z,V]$.
\end{assumption}
Assumption \ref{assn_treatment}  is similar to Assumption \ref{assn_1h}(b) and  imposes an exclusion restriction on the instrument $Z$ from the outcome equation given the unobserved heterogeneity $U^*$. Assumption \ref{assn_treatment} corresponds to Assumption A-1 of \citet{heckman_etal06rest}, which assumes $(\eta_0,\eta_1)$ is independent of $Z$ conditional on $V$.


Similar to \citet{heckman_etal06rest}, we can write the observed outcome under true treatment as
\begin{align*}
Y & = Y_0 + (Y_1 - Y_0) T^* \\
  & = \mu_0(V) + E[\eta_0|U^*,V] + \left[ \mu_1(V) - \mu_0(V) + \eta_1 - \eta_0 \right] T^* + \eta_0 - E[\eta_0|U^*,V]\\
  & = \alpha(U^*,V) + \beta(U^*,V) T^* +\ve,
\end{align*}
with defining 
\begin{align*}
\alpha(U^*,V) & := \mu_0(V) + E[\eta_0|U^*,V], \\
\beta(U^*,V) & := \mu_1(V) - \mu_0(V) + \eta_1 - \eta_0 , \\
\ve & := \eta_0 - E[\eta_0| U^*,V],
\end{align*}
where $\ve$ satisfies $E[\ve| U^*,Z,V]=0$ from Assumption \ref{assn_treatment}. Furthermore,
\[
Y_1 - Y_0 = \mu_1(V) - \mu_0(V) + \eta_1 - \eta_0 = \beta(U^*,V),
\]
holds. Then, we can identify $\alpha(U^*,V)$ and $\beta(U^*,V)$ from Proposition \ref{prop_3}.  When $E[\eta_1| U^*,V] = E[\eta_1| U^*,Z,V]$ holds, we may write $Y  = Y_1 + (Y_0 - Y_1)(1- T^*) = \mu_1(V) + \left[ \mu_0(V) - \mu_1(V) + \eta_0 - \eta_1\right] (1-T^*) + \eta_1$ and repeat the above argument. 

From Proposition \ref{prop_3} and Assumption \ref{assn_treatment}, we can identify the average treatment effect (ATE), the average treatment effect on the treated (TT), and the average treatment effect on the untreated (TUT) conditional on $V$ by taking the average of $\beta(U^*,V)$ over $U^*$ using appropriate weights as
\[
\begin{aligned} 
ATE & = E[Y_1 - Y_0|V] = \sum_{u\in \mathcal{U}} \beta(u,V) \Pr(U^*=u|V), \\
TT &= E[Y_1 - Y_0|T^*=1,V] =\sum_{u\in \mathcal{U}} \beta(u,V) \Pr(U^*=u|T^*=1,V),\\
TUT &= E[Y_1 - Y_0|T^*=0,V] =\sum_{u\in \mathcal{U}} \beta(u,V) \Pr(U^*=u|T^*=0,V),
\end{aligned}
\]
respectively, where  $\Pr(U^*=u|T^*,V)$ is identified from  $\Pr(S^*|Z,V)$. 

When Assumption \ref{assn_treatment} does not hold, Proposition \ref{prop_3}  identifies the marginal distribution of $Y_0$ and that of $Y_1$ separately, but the distribution of the individual treatment effects, $Y_1- Y_0$, is not point-identified. \citet{fanpark10et} provide a sharp bound on the distribution of the individual treatment effects given the marginal distributions of $Y_0$ and $Y_1$.

\section{Conclusion}
This paper gives new identification results for cross-sectional regression models when a binary regressor is misclassified and endogenous. Existing studies assume that the instrument used in estimation satisfies not only the standard exclusion restriction and relevance condition but also an additional condition that it is uncorrelated with misclassification errors. Some instruments in empirical applications, however, may be correlated with misclassification errors and, thus, relaxing this additional requirement is important for applications.  We show that the constant and slope parameters are identified even if a binary instrumental variable is correlated with misclassification errors when there exists a regressor that is excluded from the outcome equation but is relevant for the true unobserved regressor.

\section{Proofs}

\begin{proof}[Proof of Proposition  \ref{prop_1} ]
The proof uses eigenvalue decomposition as in \citet{anderson54pcma}, \citet{delathauwer06sjmaa}, \citet{hu08joe}, \citet{kasaharashimotsu09em}, and \citet{carrol10jnps}. Henceforth, we use $E[T]$ and $\Pr(T=1)$ exchangeably because $T$ is binary. Under Assumption \ref{assn_1},  we obtain the following decomposition of $E[Y|Z,V]$, $E[T|Z,V]$, and $E[YT|Z,V]$:
\begin{equation}\label{model2}
\begin{aligned}
1 & = \Pr(T^*=0|Z,V)+ \Pr(T^*=1|Z,V)  \\
E[Y|Z,V] & = E[Y|T^*=0,Z,V] \Pr(T^*=0|Z,V)  + E[Y|T^*=1,Z,V]\Pr(T^*=1|Z,V) \\
& = E[Y|T^*=0,V] \Pr(T^*=0|Z,V)  + E[Y|T^*=1,V]\Pr(T^*=1|Z,V),   \\
E[T|Z,V] & = E[T|T^*=0,Z,V]\Pr(T^*=0|Z,V)  + E[T|T^*=1,Z,V]\Pr(T^*=1|Z,V) \\
 & = E[T|T^*=0,Z]\Pr(T^*=0|Z,V)  + E[T|T^*=1,Z]\Pr(T^*=1|Z,V), \\
\end{aligned}
\end{equation}
where the third equality follows from Assumption \ref{assn_1}(b), and the fifth equality follows from Assumption \ref{assn_1}(d).

Write $E[YT|Z,V]$ as
\begin{equation} \label{EYT}
E[YT|Z,V] =  E[YT|T^*=0,Z,V] \Pr(T^*=0|Z,V) + E[YT|T^*=1,Z,V] \Pr(T^*=1|Z,V) .
\end{equation}
We proceed to simplify $E[YT|T^*,Z,V] $ in (\ref{EYT}). It follows from the law of iterated expectations and Assumption \ref{assn_1}(a) that 
\begin{align*}
E[YT|T^*, Z,V] & = E[ E[Y|T,T^*, Z,V]  T  |T^*, Z,V]  = E[Y|T^*, Z,V]E[T |T^*, Z,V].  
\end{align*}
Under Assumption \ref{assn_1}(b)(d),  the right hand side is written as $E[Y|T^*, V]E[T |T^*, Z]$. Substituting this to (\ref{EYT}) gives
\begin{equation}\label{model3}
\begin{aligned}
E[YT|Z,V] & =  E[Y|T^*=0,V] E[T|T^*=0,Z] \Pr(T^*=0|Z,V) \\
& \quad +  E[Y|T^*=1,V] E[T|T^*=1,Z] \Pr(T^*=1|Z,V). 
\end{aligned}
\end{equation}
For $(Z,V) \in \{0,1\}^2$, define the following matrices. First, define the matrix of observable conditional moments of $(Y,T)$ given $(Z,V)$ as
\begin{equation}\label{QZV}
Q(Z,V) := 
\begin{pmatrix}
1 & E[Y|Z,V]\\
E[T|Z,V] & E[YT|Z,V]
\end{pmatrix}.
\end{equation}
Next, define the matrices of unobservables as
\begin{equation}\label{LLLambda}
\begin{aligned}
L_Y(V) &:=
\begin{pmatrix}
1 & 1\\
E[Y|T^*=0,V] & E[Y|T^*=1,V]
\end{pmatrix},\\
L_T(Z) &:=
\begin{pmatrix}
1 & 1\\
E[T|T^*=0,Z] & E[T|T^*=1,Z]
\end{pmatrix},\\
\Lambda(Z,V) &:=
\begin{pmatrix}
\Pr(T^*=0|Z,V)&0\\
0& \Pr(T^*=1|Z,V)
\end{pmatrix}.
\end{aligned}
\end{equation}
Note that $L_Y(V)$ is invertible from Assumption \ref{assn_1}(g) and (\ref{mean_Y}) and that $L_T(Z)$ and $\Lambda(Z,V)$ are invertible from Assumption \ref{assn_1}(f)(h). We can collect (\ref{model2})--(\ref{model3}) as 
\begin{equation} \label{decomp}
Q(Z,V) = L_T(Z) \Lambda(Z,V) L_Y(V)\t.
\end{equation}
Evaluating (\ref{decomp}) for $(Z,V) \in \{0,1\}^2$ gives
\begin{align*}
Q(0,0) = L_T(0) \Lambda(0,0) L_Y(0)\t, \quad Q(0,1) = L_T(0) \Lambda(0,1) L_Y(1)\t, \\
Q(1,0) = L_T(1) \Lambda(1,0) L_Y(0)\t, \quad Q(1,1) = L_T(1) \Lambda(1,1) L_Y(1)\t. 
\end{align*}
Observe that
\begin{align*}
Q(0,0) Q(1,0)^{-1} Q(1,1) Q(0,1)^{-1} &= L_T(0) \Lambda(0,0) \Lambda(1,0)^{-1} \Lambda(1,1) \Lambda(0,1)^{-1} L_T(0)^{-1}.
\end{align*}
Defining $\tilde Q := Q(0,0) Q(1,0)^{-1} Q(1,1) Q(0,1)^{-1}$ and  $\tilde \Lambda := \Lambda(0,0) \Lambda(1,0)^{-1} \Lambda(1,1) \Lambda(0,1)^{-1}$, we can write this equation as $\tilde Q = L_T(0) \tilde \Lambda L_T(0)^{-1}$. It follows that 
\[
\tilde Q L_T(0) = L_T(0) \tilde \Lambda.
\]
From Assumption \ref{assn_2}, the eigenvalues of $\tilde Q$ are distinct. Because $\tilde \Lambda$ is diagonal and the first row of $L_T(0)$ is 1, the columns of $L_T(0)$ are identified as the eigenvectors of $\tilde Q$. Furthermore, Assumption \ref{assn_1}(f) identifies individual columns of $L_T(0)$. Similarly, we have
\begin{align}
Q(0,0)\t (Q(0,1)\t)^{-1} Q(1,1)\t (Q(1,0)\t)^{-1} &= L_Y(0) \Lambda(0,0) \Lambda(0,1)^{-1} \Lambda(1,1) \Lambda(1,0)^{-1} L_Y(0)^{-1} \nonumber \\
& = L_Y(0) \tilde \Lambda L_Y(0)^{-1}, \nonumber 
\end{align}
where the last equality holds because $\Lambda(Z,V)$ is diagonal. Consequently, $L_Y(0)$ is identified from the eigenvectors of $Q(0,0)\t (Q(0,1)\t)^{-1} Q(1,1)\t (Q(1,0)\t)^{-1}$.

Once $L_T(0)$ and $L_Y(0)$ are identified, we can identify $\Lambda(0,0)$ as $\Lambda(0,0) = L_T(0)^{-1}Q(0,0)(L_Y(0)\t)^{-1}$. $L_T(1)$ and $L_Y(1)$ are identified from the eigenvectors of $Q(1,0) Q(0,0)^{-1} Q(0,1) Q(1,1)^{-1}$ and $Q(1,1)\t (Q(1,0)\t)^{-1} Q(0,0)\t (Q(0,1)\t)^{-1}$, respectively. Further, $(\Lambda(0,1), \Lambda(1,0), \Lambda(1,1))$ is identified from $(L_T(0),L_T(1),L_Y(0),L_Y(1))$ and $(Q(0,1),Q(1,0),Q(1,1))$. Finally, $\alpha(V)$ and $\beta(V)$ are identified from the relation (\ref{alpha_beta}) and Assumption \ref{assn_1}(c). When $E[T|T^*,Z]$ does not depend on $Z$, the proof remains unchanged except that $L_T(Z)$ is replaced with $L_T$.
\end{proof}

\begin{proof}[Proof of Proposition \ref{prop_2}]
Using a similar derivation to (\ref{model2}) and (\ref{model3}), we obtain 
\begin{equation}\label{model_alt}
\begin{aligned}
1 & = \Pr(T^*=0|Z)+ \Pr(T^*=1|Z)  \\
E[Y|Z,V] & = E[Y|T^*=0,V] \Pr(T^*=0|Z)  + E[Y|T^*=1,V]\Pr(T^*=1|Z),   \\
E[T|Z,V] & = E[T|T^*=0]\Pr(T^*=0|Z)  + E[T|T^*=1]\Pr(T^*=1|Z), \\
E[YT|Z,V] & =  E[Y|T^*=0,V] E[T|T^*=0] \Pr(T^*=0|Z)\\
& \quad + E[Y|T^*=1,V] E[T|T^*=1] \Pr(T^*=1|Z) . 
\end{aligned}
\end{equation}
Define $Q(Z,V)$ and $L_Y(V)$ as in (\ref{QZV}) and (\ref{LLLambda}), and define  
\begin{align*}
 L_T  :=
\begin{pmatrix}
1 & 1\\
E[T|T^*=0] & E[T|T^*=1]
\end{pmatrix}\quad\text{and}\quad
\Lambda(Z)  :=
\begin{pmatrix}
\Pr(T^*=0|Z)&0\\
0& \Pr(T^*=1|Z)
\end{pmatrix}.
\end{align*}
Then, we  can collect (\ref{model_alt}) as 
\begin{equation} \label{decomp_alt}
Q(Z,V) = L_T \Lambda(Z) L_Y(V)\t.
\end{equation}
Observe that
\begin{align*}
Q(0,0) Q(1,0)^{-1}  &= L_T \Lambda(0) \Lambda(1)^{-1}  L_T^{-1}.
\end{align*}
From Assumption \ref{assn_3}, the eigenvalues of $Q(0,0) Q(1,0)^{-1}$ are distinct. Consequently, the columns of $L_T$ are identified as the eigenvectors of $Q(0,0) Q(1,0)^{-1}$. 

Similarly, we have
\begin{align*}
Q(0,0)\t (Q(1,0)\t)^{-1} &= L_Y(0) \Lambda(0) \Lambda(1)^{-1} L_Y(0)^{-1},
\end{align*}
and $L_Y(0)$ is identified from the eigenvectors of $Q(0,0)\t (Q(1,0)\t)^{-1}$. Once $L_T$ and $L_Y(0)$ are identified, we can identify $\Lambda(0)$ as $\Lambda(0) = L_T^{-1}Q(0,0)(L_Y(0)\t)^{-1}$, and $\Lambda(1)$ and $L_Y(1)$ are identified similarly. Finally, $\alpha(V)$ and $\beta(V)$ are identified from the relation (\ref{alpha_beta}) and Assumption \ref{assn_1}(c).
\end{proof}

\begin{proof}[Proof of Proposition \ref{prop_3}]
The proof is similar to the proof of Proposition \ref{prop_1}. First, observe that Assumption \ref{assn_1h}(a)(b) and model (\ref{model_hetero}) imply that (a) $Y$ is independent of $S$ conditional on $(S^*,Z,V)$, and (b) $Y$ is independent of $Z$ conditional on $(S^*,V)$. Therefore, under Assumption \ref{assn_1h},  using a similar argument to the proof of Proposition \ref{prop_1} gives  the following representations:  for any set $\Delta$ on the support of $Y$,
\begin{equation}\label{model2b} 
\begin{aligned}
1 & =\sum_{s\in \mathcal{S} }\Pr(S^*=s |Z,V),\\
 \Pr(Y\in \Delta|Z,V) & =\sum_{s\in \mathcal{S}}\Pr(S^*=s|Z,V) \Pr(Y\in \Delta|S^*=s,V) ,\\
 \Pr(S=s|Z,V) & =  \sum_{s\in \mathcal{S}}\Pr(S^*=s|Z,V)  \Pr(S=s|S^*=s,Z),\\
\Pr(Y\in \Delta,S=s|Z,V) &=    \sum_{s\in \mathcal{S}}\Pr(S^*=s|Z,V) \Pr(Y\in \Delta|S^*=s,V)  \Pr(S=s|S^*=s,Z).
\end{aligned}
\end{equation} 
 
Consider an event $\{Y\in \Delta_j\}$ for $j=1,\ldots,K$, where $\{\Delta_j\}_{j=1}^K$ satisfies Assumption \ref{assn_3h}.  Evaluating (\ref{model2b})   at different values of $(\Delta,s)$ given $(Z,V)$ and stack them into matrices, we have
\begin{equation} \label{decomp*}
Q(Z,V) = L_T(Z) \Lambda(Z,V) L_Y(V)\t,
\end{equation}
where
\begin{align*}
Q(Z,V) & := 
\begin{bmatrix}
1& \Pr(Y\in \Delta_1|Z,V)& \cdots &\Pr(Y\in \Delta_{K-1}|Z,V)\\
\Pr(S=s_1|Z,V)&\Pr(Y\in \Delta_1,S=s_1|Z,V)&\cdots &\Pr(Y\in \Delta_{K-1},S=s_1|Z,V)\\
\vdots & \ddots & \vdots & \vdots \\
\Pr(S=s_{K-1}|Z,V)&\Pr(Y\in \Delta_1,S=s_{K-1}|Z,V)&\cdots &\Pr(Y\in \Delta_{K-1},S=s_{K-1}|Z,V)
\end{bmatrix},\\
 L_T(Z)&:= 
\begin{bmatrix}
1&1&\cdots &1\\
\Pr(S=s_1|S^*=s_1,Z)&\Pr(S=s_1|S^*=s_2,Z)&\cdots & \Pr(S=s_1|S^*=s_{K},Z)\\
\vdots & \ddots & \vdots& \vdots  \\
\Pr(S=s_{K-1}|S^*=s_1,Z)&\Pr(S=s_{K-1}|S^*=s_2,Z)&\cdots & \Pr(S=s_{K-1}|S^*=s_{K},Z\\
\end{bmatrix},\\
 L_Y(V)&:=
\begin{bmatrix}
1&1&\cdots &1\\
\Pr(Y\in \Delta_1|S^*=s_1,V)&\Pr(Y\in \Delta_1|S^*=s_2,V)&\cdots& \Pr(Y\in \Delta_1|S^*=s_{K},V)\\
\vdots & \ddots & \vdots& \vdots  \\
\Pr(Y\in \Delta_{K-1}|S^*=s_1,V)&\Pr(Y\in \Delta_{K-1}|S^*=s_2,V)&\cdots& \Pr(Y\in \Delta_{K-1}|S^*=s_{K},V)
\end{bmatrix},\\
\Lambda(Z,V)&:=
\begin{bmatrix}
\Pr(S^*=s_1|Z,V)&0&  \cdots&0\\
0& \Pr(S^*=s_2|Z,V)& \cdots&0\\
\vdots&\vdots& \ddots &   \vdots\\ 
0&\cdots&\cdots & \Pr(S^*=s_{K}|Z,V)
\end{bmatrix}.
\end{align*}

Note that $L_T(Z)$ is nonsingular because $L_T(Z)$ is written as
\[
L_T(Z) = 
\begin{bmatrix}
1 & 1 & \cdots& \cdots &1\\
1 & 0 & \cdots & \cdots&0\\
0 &1 & 0 & \cdots  & 0 \\ 
\vdots &\ddots& \ddots & \ddots& \vdots\\ 
0&\cdots&0 & 1 & 0
\end{bmatrix} A , \quad A := 
\begin{bmatrix}
\Pr(S=s_1|S^*=s_1,Z)& \cdots & \Pr(S=s_1|S^*=s_{K},Z)\\
\vdots & \ddots & \vdots    \\
\Pr(S=s_{K}|S^*=s_1,Z)&\cdots & \Pr(S=s_{K}|S^*=s_{K},Z)\\
\end{bmatrix},
\]
and $A$ is strictly diagonally dominant by Assumption \ref{assn_1h}(d). $L_Y(V)$ and $\Lambda(Z,V)$ are nonsingular by Assumption \ref{assn_3h} and Assumption \ref{assn_1h}(f), while the diagonal elements of $\tilde \Lambda := \Lambda(0,0) \Lambda(1,0)^{-1} \Lambda(1,1) \Lambda(0,1)^{-1}$ are distinct by Assumption \ref{assn_2h}.

Then, evaluating (\ref{decomp*}) at $(Z,V)\in\{(0,0),(1,0),(0,1),(1,1)\}$ and repeating the argument in the proof of Proposition \ref{prop_1} under Assumptions \ref{assn_1h}--\ref{assn_3h} identify $L_T(Z)$, $L_Y(V)$, and $\Lambda(Z,V)$ for $(Z,V)\in\{(0,0),(1,0),(0,1),(1,1)\}$ up to the permutation of columns across different values of $S^*$. Assumption \ref{assn_1h}(d) identifies the ordering of the eigenvectors in $L_T(Z)$, which gives the identification of columns of $L_Y(V)$ and $\Lambda(Z,V)$. This identifies $\Pr(S^*|Z,V)$ and $\Pr(S|S^*,Z)$ for all $(S,S^*,Z,V)$ and $\Pr(Y\in \Delta_j|S^*,V)$ for $\{\Delta_j\}_{j=1}^K$ that satisfies Assumption \ref{assn_3h}.

For identifying $\alpha(U^*,V)$ and $\beta(U^*,V)$, observe that 
\[
E[Y|U^*,V,Z] = \alpha(U^*,V) + \beta(U^*,V)E[T^*|U^*,V,Z],
\]
because $E[\ve|U^*,V,Z]=0$ from (\ref{model_hetero}). It follows that
\begin{equation}\label{alpha_beta_hetero}
\begin{pmatrix}
E[Y|U^*,V,Z=0] \\
E[Y|U^*,V,Z=1] 
\end{pmatrix}
=
\begin{pmatrix}
1 & \Pr(T^*=1|U^*,Z=0,V) \\
1 & \Pr(T^*=1|U^*,Z=1,V) 
\end{pmatrix}
 \begin{pmatrix}
\alpha(U^*,V) \\
\beta(U^*,V)
\end{pmatrix} .
\end{equation}
We can identify $E[Y|U^*,V,Z]$ because we can derive $\Pr(Y\in \Delta_j,S^* =s|Z,V)$ from $\Pr(Y\in \Delta_j|S^*,V)$ and $\Pr(S^*|Z,V)$. We can derive $\Pr(T^*=1|U^*,Z,V)$ from $\Pr(S^*|Z,V)$. Therefore, the left hand side of (\ref{alpha_beta_hetero}) and the matrix on the right hand side of (\ref{alpha_beta_hetero}) are identified, and this matrix is invertible from Assumption \ref{assn_1h}(c). Consequently, $\alpha(U^*,V)$ and $\beta(U^*,V)$ are identified.

It remains to show the identification of $\Pr(Y\in \Delta|S^*,V)$ for any $\Delta$. For a partition $\bs{\Delta}:=\{\Delta_j\}_{j=1}^K$ that does not satisfy Assumption \ref{assn_3h}, define
\[
 L_{Y,\bs{\Delta}} (V) :=
\begin{bmatrix}
1&1&\cdots &1\\
\Pr(Y\in \Delta_1|S^*=s_1,V)&\Pr(Y\in \Delta_1|S^*=s_2,V)&\cdots& \Pr(Y\in \Delta_1|S^*=s_{K},V)\\
\vdots & \ddots & \vdots& \vdots  \\
\Pr(Y\in \Delta_{K-1}|S^*=s_1,V)&\Pr(Y\in \Delta_{K-1}|S^*=s_2,V)&\cdots& \Pr(Y\in \Delta_{K-1}|S^*=s_{K},V)
\end{bmatrix}
\]
and 
$Q_{Y,\bs{\Delta}}  (Z,V) : = L_T(Z) \Lambda(Z,V)  L_{Y,\bs{\Delta}}(V)$. Given the identification of $L_T(Z)$ and $\Lambda(Z,V)$, we may identify $L_{Y,\bs{\Delta}}(V)$ as  $ L_{Y,\bs{\Delta}}(V) =( \Lambda(Z,V) )^{-1}  ( L_T(Z) )^{-1} Q_{Y,\bs{\Delta}}  (Z,V)$. Because a partition $\bs{\Delta}:=\{\Delta_j\}_{j=1}^K$ is arbitrary, $\Pr(Y\in \Delta|S^*,V)$ is identified for any $\Delta$ for all $(S^*,V)$.
\end{proof}

\bibliography{mixture}

\end{document}